\documentclass[12pt,a4paper]{amsart}

\date{January 23, 2013}

\newtheorem{theorem}{Theorem}

\def\cz{\mathbb{C}} % komplexe Zahlen
\def\rz{\mathbb{R}} % reelle Zahlen
\def\bp{\mathbf{p}}
\def\bx{\mathbf{x}}
\def\by{\mathbf{y}}
\def\bA{\mathbf{A}}

\def\cE{\mathcal{E}}

\def\gA{\mathfrak{A}}
\def\gH{\mathfrak{H}}

\def\rd{\mathrm{d}}
\def\ri{\mathrm{i}}

\def\sgn{\mathrm{sgn}}

\begin{document}
\title[Bound on 2-$d$-Excess Charge]{On the Maximal Excess Charge of the Chandrasekhar-Coulomb Hamiltonian in Two Dimensions}

\author[M. Handrek]{Michael Handrek}
\author[H. Siedentop]{Heinz Siedentop}
\address{Mathematisches Institut\\
 Ludwig-Maximilians-Universit\"at M\"unchen\\
 Theresienstra\ss e 39\\ 80333 M\"unchen\\ Germany}
 \email{handrek@math.lmu.de  \textit{and} h.s@lmu.de}

\maketitle

\begin{abstract}
  We show that for the straightforward quantized relativistic Coulomb
  Hamiltonian of a two-dimensional atom -- or the corresponding
  magnetic quantum dot -- the maximal number of electrons does not
  exceed twice the nuclear charge. The result is then generalized to
  the presence of external magnetic fields and atomic
  Hamiltonians. This is based on the positivity of $$|\bx| T(\bp) +
  T(\bp) |\bx| $$ which -- in two dimensions -- is false for the
  non-relativistic case $T(\bp) = \bp^2/2$, but is proven in this paper
  for $T(\bp) = |\bp|$, i.e., the ultra-relativistic kinetic energy.
 \end{abstract}

\section{Introduction\label{ein}}

The energy of two-dimensional quantum systems interacting via three
dimensional Coulomb-like potentials like graphene in a background
magnetic field given by a vector potential $\gA$ is given by the
quadratic form
\begin{equation}
  \label{eq:form}
  \cE[\psi] = (\psi,W_n\psi)
\end{equation}
where $\psi$ is a sufficiently smooth element from the 
Hilbert space $\gH_\gA$ built from states above the Fermi level:
\begin{equation}
  \label{Einteilchen}
  \psi\in\gH_\gA:=\Lambda^+[L^2(\rz^2:\cz^2)] \;, \quad \Lambda^+:= \chi_{(0,\infty)}(D_\gA). 
\end{equation}
(Note that this model is physically relevant for large gaps and
relatively small interactions. For strong interactions and small gaps
a field theoretic description allowing for particle-hole creation and
annihilation is required (Paananen and Egger
\cite{PaananenEgger2011}).)  The operator
\begin{equation} \label{physikhamilton}
 W_N:= \sum_{n=1}^N \left(T_{\gA,n} +e\boldsymbol\sigma_n\cdot\bA(\bx_n) 
-e\varphi(\bx_n) \right) + \sum_{1\leq m<n\leq N}\frac{e^2}{|\bx_m-\bx_n|}
\end{equation}
is the multi-particle Weyl operator (massless Dirac operator) with
kinetic energy $T_\gA:=\boldsymbol\sigma\cdot(\bp+e\gA)$ where $-e$ is
the charge of the electron. The Hamiltonian is defined via the
quadratic form. The Hamiltonian has been used to describe
multi-particle effects of graphene quantum dots (see Egger et
al \cite{Eggeretal2010} where its basic mathematical properties have
been discussed). The electro-magnetic potentials $\varphi$ and $\gA$
are defining the quantum dot.

To be concrete we mention that the choice used by Egger et
al \cite{Eggeretal2010} (see also Paananen et
al \cite{Paananenetal2011}) would be allowed, namely to take
$\gA= \gA_0 + \bA_0$ and $\bA=0$ 
with
\begin{equation}
  \label{hintergrund}
  \gA_0(\bx) = \tfrac B2 
\begin{pmatrix} 
-x_2\\
x_1
\end{pmatrix},
\end{equation}
i.e., we have a homogeneous magnetic field of strength $B$ orthogonal
to the $x_1$-$x_2$-plane, and
\begin{equation}
  \label{magdot}
  \bA_0(\bx) := -\frac{ BR^2}{2|\bx|^2}
  \begin{pmatrix}
    -x_2\\x_1  \end{pmatrix}
  \begin{cases}
    {|\bx|^2\over R^2} & |\bx| \leq R\\
    1 & |\bx| > R,
  \end{cases}
\end{equation}
the potential that eliminates the magnetic field in a circle of radius
$R$ but leaves the field unchanged outside.  In total this is a
homogeneous magnetic field with a cavity of radius $R$ around the
origin.  As electric field we could choose $\varphi(\bx)=eZ/|\bx|$, i.e., the
potential of a defect atom placed at the origin.

An alternative approach would be to take $\gA=\gA_0$, i.e., define the
vacuum with respect to the homogeneous magnetic field, and to pick
$\bA=\bA_0$. Although for weak fields these two will yield similar
results they are not identical. Nevertheless, by the variational
principle for eigenvalues in gaps \cite{GriesemerSiedentop1999}, the
latter will bound the former from below. In addition it allows for a
more direct treatment of the excess charge problem. Because of this,
we will direct our attention on the second choice.

The corresponding quadratic form on
$\bigwedge_{n=1}^N \Lambda^+[C^\infty_0(\rz^2:\cz^2)]$ is bounded from
below, if and only if
$$Z\leq  \left({\Gamma(\tfrac14)^4\over8\pi^2}+ {8\pi^2\over\Gamma(\tfrac14)}\right)^{-1}
$$ 
independently of the strength of the field $\gA$
\cite{MaierSiedentop2012}. Egger et al \cite{Eggeretal2010} studied
also numerically how many electrons a quantum dot is able to bind in
the context of a mean-field model. However the question of bounding
the total number of electrons localized in the quantum dot was left
unanswered. In the sequel we will address this question, however not
for the no-pair Hamiltonian defined by the quadratic form of $W_N$ in
$\bigwedge_{n=1}^N\gH_\gA$ but for the corresponding Chandrasekhar
type operator $C_{\bA,\varphi,N}$, i.e., $W_N$ with $T_\gA$ replaced
by $|T_\gA|$ self-adjointly realized in $L^2(\rz^{3N}:\cz^{2^N})$ with
domain $H^1(\rz^{3N}:\cz^{2^N})$. Again by the variational principle,
the eigenvalues of this operator bound the eigenvalues of the no-pair
operator from below, since the quadratic form of the no-pair operator
is just a restriction of the Chandrasekhar case.

An a priori bound on the maximal number of electrons that can be bound
has been derived by Lieb \cite{Lieb1984} for relativistic and
non-relativistic Coulomb Hamiltonians in three dimensions. The idea of
the proof is to multiply the Schr\"odinger equation by $|\bx_1| \psi(x)$
and to integrate. It yields -- in the atomic case -- the bound $N <
2Z+1$ and rests on the inequality
\begin{equation}
  \label{grundl}
  0\leq |\bx| T + T |\bx|.
\end{equation}
For the non-relativistic kinetic $T=\bp^2/2$ \eqref{grund} is equivalent
to a Hardy inequality which is true in three dimensions but false in
two dimension. The case of the relativistic kinetic energy $T=|\bp|$
is reduced by Lieb \cite{Lieb1984} to the non-relativistic case. In
other words, it is not clear a priori whether \eqref{grund} holds. The
purpose of this paper is to show this inequality and apply it to the Hamiltonian 
$C_{\bA,\varphi,N}$ of the quantum dot.

Before we actually do this we remark:
\begin{itemize}
\item In three dimensions Nam\cite{Nam2012} improved Lieb's result for
  the non-relativistic Schr\"odinger operator and $Z\geq6$. He could
  show the bounds $N< 1.22 Z+3Z^{\tfrac1{3}}$ using similar ideas. 
\item In two dimension, the above failure of positivity of
  \eqref{grund} can be controlled using an idea of Seiringer
  \cite{Seiringer2001} and Nam and Solovej \cite{Nametal2012}: Since
  the eigenvalues of the two-dimensional hydrogen atoms
  (two-dimensional Kepler problem) are $-Z^2/2(n+1/2)^2$ (Fl\"ugge and
  Marschall \cite[Problem 24]{FluggeMarschall1965} each of
  multiplicity $2n+1$ ($e=1$). This gives $(\psi, |\bx_1|^{-1}\psi) \leq
  4\log(Z^{1/2})+10$ for the ground state of the $N$ particle system
  which in turn yields
  \begin{equation}
    \label{eq:2d}
    N\leq 2 Z  + \log(Z^{1/2}) + \frac{7}{2}.
  \end{equation}
 \end{itemize}

\section{Main Inequality\label{s:1}}
Our basic result is
\begin{theorem}
  Assume $ \gA \in L^2_{\mathrm{loc}}(\rz^2:\cz^2)$, $\bp:=
  -\ri \nabla$, and $T_{\mathbf{\gA}}:= |\bp+\gA|$, then 
  \begin{equation}
    \label{grund}
    |\bx| T_\gA
  + T_\gA |\bx|\geq0
  \end{equation}
  on $C_0^\infty(\rz^d)$.
\end{theorem}

\begin{proof}
  By the diamagnetic inequality (Theorem \ref{diamagnetic})
  \begin{equation}
    \label{diamagnetic}
  (\eta, |\bp| |\phi|) \leq \Re (\eta, \sgn(\phi)
  |\bp+\gA|\phi)
  \end{equation}
  we see that $ |\bx||\bp| + |\bp||\bx| \geq 0$ implies
  $|\bx||\bp+\gA| + |\bp+\gA||\bx| \geq0$. In other words, it suffices
  to prove positivity without the background magnetic field.
  Moreover: it suffices to consider $\phi\geq0$.

  Next we recall the following representation of $|\bp|$ (Lieb and Yau
  \cite{LiebYau1988}) in position space
  \begin{equation}
    \label{eq:betrag}
    (\psi,|\bp|\psi) = \alpha_d \int_{\rz^d}\rd \bx \int_{\rz^d}\rd \by 
    {(\overline{\psi(\bx)}-\overline{\psi(\by)})(\psi(\bx)-\psi(\by))\over |\bx-\by|^{d+1}}
  \end{equation}
 with $\alpha_d=\Gamma(\tfrac{d+1}2)/(2\pi^{d+1\over2})$.
  By polarization of \eqref{eq:betrag} Inequality \eqref{grund} is then equivalent to 
  \begin{equation}
    \label{eq:4}
    0\leq t:=\Re\int_{\rz^d}\rd \bx \int_{\rz^d}\rd \by 
    {(\overline{\psi(\bx)}-\overline{\psi(\by)})(|\bx|\psi(\bx)-|\by|\psi(\by))\over|\bx-\by|^{d+1}}
  \end{equation}
  where we restrict to non-negative $\psi$ because of the above remark.
Now, setting $\psi= g / |.|^{d/2}$ and regularizing to avoid the singularity at $\bx=\by$ we get
\begin{align}
  \label{bound}
  \begin{split}
    t =&\lim_{\epsilon\to 0} \Re \int_{\rz^d} \rd \bx \int_{\rz^d} \rd \by  {{|\bx|g(\bx)^2\over|\bx|^d}+{|\by|g(\by)^2\over|\by|^d}-g(\bx)g(\by){|\by| + |\bx|\over|\bx|^{d/2} |\by|^{d/2} } \over |\bx-\by|^{d+1}+2^{\frac{d+1}{2}}\epsilon(|\bx|^{d+1}+|\by|^{d+1})}\\
    =&\lim_{\epsilon\to 0} \int_{\rz^d} \rd \bx
    {g(\bx)^2\over|\bx|^d}\int_{\rz^d} \rd \by {2 - |\by|^{1-d/2} -|\by|^{-d/2}
      \over|\mathfrak{e}-\by|^{d+1}+2^{\frac{d+1}{2}}\epsilon(1+|\by|^{d+1})}\\
    &+ \frac12\int_{\rz^d} \rd \bx \int_{\rz^d} \rd \by {(|\bx|+|\by|)(g(\bx)-g(\by))^2\over |\bx|^{d/2}|\bx-\by|^{d+1}|\by|^{d/2}}\\
    \geq& \int_{\rz^d} \rd \bx {g(\bx)^2\over|\bx|^d} \int_0^\infty {\rd r
      r^d\over r}
    (2 -r^{1-d/2} - r^{-d/2} )(2r)^{-{d+1\over2}}\\
    &\cdot \int_{\mathbb{S}^{d-1}}{\rd \omega\over
      ({r+r^{-1}\over2} - \omega \mathfrak{e})^{{d+1\over2}}+\epsilon(\frac1{r}^{{d+1\over2}}+r^{{d+1\over2}})}\\
    \geq& 2^{-{d+1\over2}}\int_{\rz^d} \rd \bx {g(\bx)^2\over|\bx|^d}
    \int_0^\infty {\rd r \over r}
    (2r^{d-1\over2} -r^{1/2} - r^{-1/2}) \\
    &\cdot \int_{\mathbb{S}^{d-1}}{\rd \omega\over
      ({r+r^{-1}\over2} - \omega \mathfrak{e})^{{d+1\over2}}+\epsilon(\frac1{r}^{{d+1\over2}}+r^{{d+1\over2}})}\\
    \geq& 2^{-{d+1\over2}}\int_{\rz^d} \rd \bx {g(\bx)^2\over|\bx|^d}
    \int_0^\infty {\rd r\over r}
    \underbrace{[r^{d-1\over 2} + r^{-{d-1\over2}} - (r^{1\over2}+r^{-{1\over2}})]}_{\geq0} \\
    &\cdot \int_{\mathbb{S}^{d-1}}{\rd \omega\over
      ({r+r^{-1}\over2} - \omega \mathfrak{e})^{{d+1\over2}}+\epsilon(\frac1{r}^{{d+1\over2}}+r^{{d+1\over2}})}
    \geq 0
 \end{split}
\end{align}
where $\mathfrak{e}$ is any unit vector in $\rz^3$. The positivity of
the bracket follows from the fact that the function $f(\alpha) :=
r^\alpha +r^{-\alpha}$ is strictly monotone increasing for positive
$r$.
\end{proof}
We remark that the above proof also shows that $d=2$ is borderline for
positivity. In fact,

$$ |\bx||\bp| + |\bp||\bx| \geq 2\alpha_d\gamma_d$$
where
\begin{multline}
	\gamma_d=2^{-{d-1\over2}} \int_0^1 {\rd r\over r}
  [r^{d-1\over 2} + r^{-{d-1\over2}} - (r^{1\over2}+r^{-{1\over2}})]\int_{\mathbb{S}^{d-1}}{\rd \omega\over
    ({r+r^{-1}\over2} - \omega \mathfrak{e})^{{d+1\over2}}}
\end{multline}
which changes sign at $d=2$ whereas $|\bx||\bp|+|\bp||\bx|\geq 1$ for $d=3$.

\section{Application to $2d$ Quantum Dots\label{s2}}

In this section we consider a $2d$ quantum dot given by the
Hamiltonian
\begin{equation} \label{model}
  C_{A,\varphi,N}:=\sum_{n=1}^N \left[ |\bp+\gA|_n +e\boldsymbol\sigma_n\cdot\bA(\bx_n) 
-e\varphi(\bx_n) \right] + \sum_{1\leq m<n\leq N}\frac{e^2}{|\bx_m-\bx_n|},
\end{equation}
a simplified model of \eqref{physikhamilton}, self-adjointly realized
in $\gH_N := \bigwedge_{n=1}^NL^2(\rz^2:\cz^2)$. Here $\gA$ is the
background magnetic field \eqref{hintergrund} and $\bA$ is the
magnetic field defining the dot. Furthermore, we allow for an
attractive essentially spherically symmetric attractive potential. We
have
\begin{theorem}
  \label{th:ueberschuss}
  Assume $\gA\in L^2_{\mathrm{loc}}(\rz^2:\rz^2)$ and $|\bA(\bx)|\leq
  e\delta/|\bx|$, $\varphi(\bx)\leq eZ/|\bx|$, $Z \in
  [0,\kappa_k]$. Assume that $C_{\bA,\varphi,N}$ has a ground state
  with ground state energy $E_N$ below the saturation threshold,
  i.e., $E_N< E_{N-1}$. Then
$$N < 2(\delta+Z) +1.$$
\end{theorem}
Note that our bound -- in the absence of an electric potential --
grows linearly in the missing magnetic flux $\mu$ which, in the
case of \eqref{magdot}, equals $\mu =e\delta= R^2Be/2$.

Equipped with Inequality \eqref{grund} the proof follows now the lines
of Lieb's original proof.

\begin{proof}
  Assume that there is ground state $\psi$ of $C_{\bA,\varphi,N}$, i.e.,
  \begin{equation}
  \label{eq:eigen}
  C_{\bA,\varphi,N}\psi= E_N\psi
\end{equation}
with $E_N$.  We begin by singling out the first coordinate and
multiplying \eqref{eq:eigen} by $|\bx_1|\psi$ and obtain using
\eqref{grund} and a standard density argument (approximating $\psi$ by
$C^\infty_0$ functions)
  \begin{multline}\label{16}
     e^2(|\bx_1|\psi,[-(\delta+Z)/|\bx_1|+ \sum_{n\neq1}
    {1\over|\bx_1-\bx_n|}]\psi)\\
    < (|\bx_1|\psi,  [C_{\bA,\varphi,1} - e^2\sum_{n\neq1}
    |\bx_1-\bx_n|^{-1}]\psi + 1\otimes
    (C_{\bA,\varphi,N-1} - E_{N-1})\psi)\\
    =(E_N-E_{N-1})(|\bx_1|\psi,\psi) < 0.
  \end{multline}
  Repeating the same argument for all the other coordinates
  $\bx_2,...,\bx_n$ and summing the result, gives
  \begin{equation}
    -(\delta+Z)N + \frac12 \sum_{m,n=1\atop m\neq n}^N{|\bx_m|+|\bx_n|\over |\bx_m-\bx_n|}<0.
  \end{equation}
  Thus the triangle inequality implies 
  $$ -(\delta+Z)N < {N^2-N\over 2}$$
  which is the claimed result.
\end{proof}

\appendix
\section{Auxiliary Results\label{a1}}
\begin{theorem}
  For $\phi, \eta\in H^1(\rz^d)$, $\eta\geq0$, $m\in\rz_+$, $T_m(\bp)
  = \sqrt{\bp^2+m^2} - m$, $\gA\in L^2_{\mathrm{loc}}(\rz^3)$ and
  $\bp=-\ri \nabla$ we have
  \begin{equation}
    \label{diamagnetic}
  (\eta, T_m(\bp) |\phi|) \leq \Re (\eta, \sgn(\phi)
  T_m(\bp+\gA)\phi).
  \end{equation}
\end{theorem}
This is Kato's inequality for the Chandrasekhar operator (Simon
\cite{Simon1979K}, K\"onenberg et al \cite{Konenbergetal2012}).

\textit{Acknowledgment:} We thank Reinhold Egger for helpful comments
on the first draft of the manuscript. Partial financial support of the
DFG through SFB-TR 12 is acknowledged.  

%\bibliographystyle{plain}
%\bibliography{coulomb}

\begin{thebibliography}{10}

\bibitem{Eggeretal2010}
Reinhold Egger, Alessandro~De Martino, Heinz Siedentop, and Edgardo Stockmeyer.
\newblock Multiparticle equations for interacting {D}irac fermions in
  magnetically confined graphene quantum dots.
\newblock {\em Journal of Physics A: Mathematical and Theoretical},
  43(21):215202, 2010.

\bibitem{FluggeMarschall1965}
Siegfried Fl\"ugge and Hans Marschall.
\newblock {\em Rechenmethoden der Quantentheorie}, volume~6 of {\em
  Heidelberger Taschenb\"ucher}.
\newblock Springer-Verlag, Berlin, 3 edition, 1965.

\bibitem{GriesemerSiedentop1999}
Marcel Griesemer and Heinz Siedentop.
\newblock A minimax principle for the eigenvalues in spectral gaps.
\newblock {\em J. London Math. Soc. (2)}, 60(2):490--500, 1999.

\bibitem{Konenbergetal2012}
Martin K\"onenberg, Oliver Matte, and Edgardo Stockmeyer.
\newblock Hydrogen-like atoms in relativistic {QED}.
\newblock In {\em Complex Quantum Systems}, In press.

\bibitem{Lieb1984}
Elliott~H. Lieb.
\newblock Bound on the maximum negative ionization of atoms and molecules.
\newblock {\em Phys. Rev. A}, 29(6):3018--3028, June 1984.

\bibitem{LiebYau1988}
Elliott~H.\ Lieb and Horng-Tzer Yau.
\newblock The stability and instability of relativistic matter.
\newblock {\em Comm.\ Math.\ Phys.}, 118:177--213, 1988.

\bibitem{MaierSiedentop2012}
Thomas Maier and Heinz Siedentop.
\newblock Stability of impurities with {C}oulomb potential in graphene with
  homogeneous magnetic field.
\newblock {\em J. Math. Phys.}, 53:095207, July 2012.

\bibitem{Nam2012}
Phan Nam.
\newblock New bounds on the maximum ionization of atoms.
\newblock {\em Communications in Mathematical Physics}, 312:427--445, 2012.
\newblock 10.1007/s00220-012-1479-y.

\bibitem{Nametal2012}
Phan Nam, Fabian Portmann, and Jan Solovej.
\newblock Asymptotics for two-dimensional atoms.
\newblock {\em Annales Henri Poincare}, 13:333--362, 2012.
\newblock 10.1007/s00023-011-0123-2.

\bibitem{PaananenEgger2011}
Tomi Paananen and Reinhold Egger.
\newblock Finite-size version of the excitonic instability in graphene quantum
  dots.
\newblock {\em Phys. Rev. B}, 84:155456, Oct 2011.

\bibitem{Paananenetal2011}
Tomi Paananen, Reinhold Egger, and Heinz Siedentop.
\newblock Signatures of {W}igner molecule formation in interacting {D}irac
  fermion quantum dots.
\newblock {\em Phys. Rev. B}, 83(8):085409, Feb 2011.

\bibitem{Seiringer2001}
Robert Seiringer.
\newblock On the maximal ionizaton of atoms in strong magnetic fields.
\newblock {\em J. Phys. A.: Math. Gen.}, 34:1943--1948, 2001.

\bibitem{Simon1979K}
Barry Simon.
\newblock {K}ato's inequality and the comparison of semigroups.
\newblock {\em J. Funct. Anal.}, 32(1):97--101, 1979.

\end{thebibliography}

\end{document}